\documentclass[twoside,12pt,reqno]{amsart}
\usepackage{amsmath,amsthm, amsfonts,amssymb}
\usepackage{color,soul}
\usepackage{enumerate}
\usepackage[dvipsnames]{xcolor}
\usepackage[all]{xy}
\usepackage{graphicx}
\usepackage{indentfirst}
\usepackage{bm}
\usepackage{mathrsfs}
\usepackage{latexsym}
\usepackage{hyperref}
\hypersetup{
    colorlinks = true,
    linkcolor = blue,
    urlcolor = blue,
    citecolor = blue,
    anchorcolor = black,
    pdftoolbar = true
    linkbordercolor = {white},
}

\usepackage{microtype}	
\usepackage{bookmark}

\usepackage{tikz}[2010/10/13]

\newcommand{\lrp}[1]{\left(#1\right)}

\newcommand{\norm}[1]{\Vert #1\Vert}
\newcommand{\FS}{\mathfrak{F}_{\text{{s}}}}

\newcommand{\C}{\mathscr{C}}

\newcommand{\Hil}{\mathcal{H}}
\newcommand{\Hild}{\widehat{\Hil}}
\renewcommand{\geq}{\geqslant}
\renewcommand{\leq}{\leqslant}

\newcommand{\be}{\begin{equation}}
\newcommand{\ee}{\end{equation}}
\newcommand{\beq}{\begin{eqnarray}}
\newcommand{\eeq}{\end{eqnarray}}
\newcommand{\beqs}{\begin{eqnarray*}}
\newcommand{\eeqs}{\end{eqnarray*}}
\newcommand{\lra}[1]{\left\langle #1 \right\rangle}

\newcommand{\RR}[1]{{\bf R$_{\bf #1}$}}
\newcommand{\Sim}[1]{{\bf S$_{\bf #1}$}}

\tikzstyle WL=[line width=3pt,opacity=1.0]
\newcommand{\drawWL}[3]
{
    \draw[white,WL]  (#2) -- (#3);
    \draw[#1] (#2) -- (#3);
}

\newcommand{\fqudit}[5]{
\node at (#1-#3,#2+#4) {\size{$#5$}};
\draw (#1,#2) --(#1,#2+#4)  arc (0:-180:#3) -- (#1-#3-#3,#2);
}

\newcommand{\braida}[4]
{
\draw (#1+2*#3,0+#2)--(#1+2*#3, 8*#4+#2);
\draw (#1+8*#3,0+#2)--(#1+8*#3, 8*#4+#2);
\draw (#1+0*#3,8*#4+#2)--(#1+1.5*#3,6.5*#4+#2);
\draw (#1+2.5*#3,5.5*#4+#2)-- (#1+3.5*#3,4.5*#4+#2);
\draw (#1+4.5*#3,3.5*#4+#2)-- (#1+6*#3,2*#4+#2) --(#1+6*#3,0*#4+#2);
\draw (#1+0,0*#4+#2)--(#1+1.5*#3,1.5*#4+#2);
\draw (#1+2.5*#3,2.5*#4+#2)--(#1+6*#3,6*#4+#2) --(#1+6*#3,8*#4+#2);
}

\newcommand{\braidb}[4]
{
\draw (#1+2*#3,-0+#2)--(#1+2*#3,- 8*#4+#2);
\draw (#1+8*#3,-0+#2)--(#1+8*#3,- 8*#4+#2);
\draw (#1+0*#3,-8*#4+#2)--(#1+1.5*#3,-6.5*#4+#2);
\draw (#1+2.5*#3,-5.5*#4+#2)-- (#1+3.5*#3,-4.5*#4+#2);
\draw (#1+4.5*#3,-3.5*#4+#2)-- (#1+#1,-2*#4+#2) --(#1+#1,-0*#4+#2);
\draw (#1+0,-0*#4+#2)--(#1+1.5*#3,-1.5*#4+#2);
\draw (#1+2.5*#3,-2.5*#4+#2)--(#1+#1,-6*#4+#2) --(#1+#1,-8*#4+#2);
}

\newcommand{\fgraphd}[2]
{\tikz{
\braida{0}{0}{1\mn}{1\nn}
\braidb{6\mn}{16\nn}{1\mn}{1\nn}
\fqudit{12\mn}{16\nn}{-1\mn}{0\nn}{}
\fqudit{6\mn}{16\nn}{-1\mn}{1\nn}{}
\fqudit{0\mn}{8\nn}{-1\mn}{10\nn}{i_A}
\fmeasure{6\mn}{0\nn}{-1\mn}{1\nn}{-m_1}
\node at (13\mn,-4\nn) {\size{$m_1$}};
\draw (12\mn,8\nn) -- (12\mn,-6\nn);
\draw (14\mn,8\nn) -- (14\mn,-8\nn);
\draw (13\mn,-6\nn) --(11\mn,-6\nn) -- (11\mn,-8\nn) --(13\mn,-8\nn);
\draw (15\mn,-6\nn) --(21\mn,-6\nn) -- (21\mn,-8\nn) --(15\mn,-8\nn);
\node at (12\mn, -7\nn) {$T$};
\draw (16\mn,-6\nn) -- (16\mn,-4\nn);
\node at (18\mn,-5\nn) {$\cdots$};
\draw (20\mn,-6\nn) -- (20\mn,-4\nn);
\draw (16\mn,-8\nn) -- (16\mn,-10\nn);
\node at (18\mn,-9\nn) {$\cdots$};
\draw (20\mn,-8\nn) -- (20\mn,-10\nn);
\fmeasure{12\mn}{-8\nn}{-1\mn}{2\nn}{-m_2}
\node at (1\mn,-13\nn) {\size{$m_2$}};
\draw (0,0) -- (0,-14\nn);
\draw (2\mn,0) -- (2\mn,-14\nn);
}}

\newcommand{\mn}{/-3}
\newcommand{\nn}{/3}

\newcommand{\fmeasure}[5]{
\node at (#1-#3,#2-#4) {\size{$#5$}};
\draw (#1,#2) --(#1,#2-#4)  arc (0:180:#3) -- (#1-#3-#3,#2);
}

\theoremstyle{plain}
\newtheorem{theorem}{Theorem~}[section]

\newtheorem{definition}[theorem]{Definition~}

\newtheorem*{remark}{Remark}

\newcommand{\size}[1]{\fontsize{4 pt}{\baselineskip}\selectfont{#1}}

\title{A Mathematical Picture Language Project}
\author{Arthur Jaffe$^{*}$}
\email{Arthur\_Jaffe@harvard.edu}
\author{Zhengwei Liu$^{+*}$}
\email{liuzhengwei@mail.tsinghua.edu.cn}
\address[*]{Harvard University, Cambridge, MA 02138, USA}
\address[+]{Tsinghua University, Beijing, China}

\thanks
{\noindent Prepared in association with the presentation on December 29, 2018 of the  ICCM Award for the paper written by the authors:
{\it Mathematical Picture Language Project''}~\cite{JL-Pic}.
 This work was supported under Grant TRT 0159 from the Templeton Religion Trust. {AJ was supported in part by ARO Grants W911NF-19-1-0302 and  W911NF-20-1-0082.} {ZL was supported in part by Grant 100301004 from Tsinghua University.} }

\begin{document}
\maketitle

\begin{abstract}
The mathematical picture language project that we began in 2016 has already yielded interesting results.  We also point out areas of mathematics and physics where we hope that it will prove useful in the future.   
\end{abstract}

\section{Pictures}
Pictures in mathematics flourished since the time of Euclid, who lived in Greece during the $4^{\rm th}$-century BC.  About the same time in China,  Chuang-Tzu promoted the philosophy of replacing words by images.  In modern mathematics pictures are pervasive, including  Feynman diagrams, category theory, and planar algebra.  

Category theory, introduced by Eilenberg and MacLane~\cite{Eilenberg-MacLane}, developed into a picture language, with transformations that led to mathematical pictures.   Feynman's use of diagrams to represent polynomial interactions in quantum many-body processes or quantum field theory gives an  early illustration of these concepts in physics~\cite{Feynman}. Extra structure on category pictures arose in Jones' theory of 
link invariants~\cite{Jones85,Jones87} and planar algebras~\cite{JonesPA}, and the general framework of topological quantum field theory \cite{Witten88, Atiyah88}.

Although pictures in mathematics are far from new,  we believe we have gained new insights into mathematics by using pictures somewhat differently from in the past. This led to our current ``mathematical picture language project.''  We have been able to prove some interesting theoretical results about pictures, and apply them to different situations.  

One aspect of our approach which we find important is the mathematical analysis of pictures, and how one might  formulate a theory complementing the study of the use of pictures in  topology and geometry.  The analytic aspect of pictures in topological quantum field theory (TQFT) is a less-developed area than its topological  and algebraic aspects, and  has great  potential for future advances. {We also incorporate many insights from using a pictorial  Fourier transformation~\cite{QFA}.}

Another notion we emphasize is proof through pictures.  We wish to distinguish focusing on general mathematical properties of pictures,  from using pictures in a particular concrete mathematical model.  In other words, we aim to distinguish between the notion of the properties of a picture language {\bf L} on the one hand, from its use through a simulation {\bf S} to model a particular  reality {\bf R} by mathematics.  This distinction between {\bf L} and {\bf R} parallels the distinction in linguistics between \textit{syntax} and \textit{semantics}.

We propose that it is interesting to prove a  result about the language  {\bf L}, and thereby through simulation ensures results in {\bf R}. One can use a single picture language {\bf L}  to simulate several different mathematical areas. In fact a theorem in {\bf L} can ensure different theorems in different mathematical subjects \RR1, \RR2, etc., as a consequence of different simulations \Sim1, \Sim2, etc.  Different configurations of the hands of the clock reveal the interrelation between picture proofs for seemingly unrelated mathematical results.
It is also important to distinguish between two types of concepts in {\bf R} that we simulate by a given {\bf S}. These may be {\it real} concepts, or they may be {\it virtual}. This distinction is not absolute, but depends on what language and simulation one considers. 

\subsection{Virtual Processes}
The mathematical interpretation of Feynman diagrams 
arises from the pictorial representation of Gaussian moments.  These occur naturally in real, scalar quantum field theory with a quadratic potential,  through analytic continuation: one replaces  the Schr\"odinger equation with  the parabolic heat equation, the solution to which is described by a Gaussian measure on the space of random fields.  The moments of this measure of degree $2n$ are the sums of  $(2n-1)!!$ terms, each represented as a Feynman diagram. 

Perturbation of the Gaussian by a  cubic potential leads to the physical interpretation by Feynman and Wheeler, and describes each diagram as a trajectory of particles that interact at the vertices of the diagram. 
An elementary Feynman diagram for the scattering of two particles of mass $m$ with energy-momentum $p,q$  colliding with and exchanging a third particle with energy-momentum $k$ is illustrated in Figure \ref{FeynmanDiagram}:
\newpage
\[
\tikz{
\draw (-.2,-1)--(0,-.1);
\draw (0,-.2)--(-.1,1);
\draw (1.2,-1)--(1,.1);
\draw (1,.1)--(1.2,1);
\draw (0,-.1)--(1,.1);
\node at (-.4,-.9) {$\scriptstyle p$};
\node at (1.45,-.9) {$\scriptstyle q$};
\node at (-.3,-.9+1.9) {$\scriptstyle p'$};
\node at (1.45,-.9+1.9) {$\scriptstyle q'$};
\node at (.5,.2) {$\scriptstyle k$};
\node at (0,-.1) {$\scriptscriptstyle\bullet$};
\node at (1,.1) {$\scriptscriptstyle\bullet$};
}
\]
\vskip -.79cm
\begin{figure}[h]
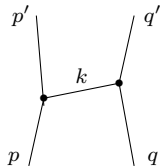

\caption{\small Virtual particle with energy-momentum $k$ in a Feynman  diagram arising from a cubic perturbation.\label{FeynmanDiagram}}
\end{figure}

After inverse analytic continuation to physical time, the energy-momentum $4$-vector $p=(E,\vec p)$ of a real, physical particle with mass-$m$ particle lies on the energy-momentum hyperboloid $E^{2}-{\vec p}^{\,2}=m^{2}$ with $E\geq 0$.  One denotes this identity as the Minkowski inner product $\lra{p,p}_{M}=m^{2}$, and one identifies the particle with its energy-momentum~$p$.  

{One can represent the collision of two particles with energy-momentum $p$ and $q$ respectively, by the exchange of a particle with energy-mo{\-}mentum~$k$. We illustrate this by the  Feynman diagram in Figure~\ref{FeynmanDiagram}. Our point is that if $p$ and $q$ are energy-momentum vectors of two  real particles of mass $m$,  and if energy-momentum is conserved at each vertex of the Feynman diagram, i.e. $p=p'+k$ and $q=q'-k$,  then the energy-momentum $k$ cannot represent another real particle. In other words,  the vector $k$  cannot both be real and satisfy $\lra{k,k}_{M}=\mu^{2}$ for real $\mu$.  Thus one says the Feynman diagram picture of scattering requires the exchange of a ``virtual'' particle, with energy-momentum $k$.    Other corrections to this elementary mediation involve the exchange of many virtual particles.}

 While some physicists would prefer to suppress such virtual effects, these concepts are very useful in many problems---especially when they illustrate symmetries that otherwise would  remain hidden. This illustration of a virtual effect illustrates how  virtual effects run through much of modern mathematics: from Fourier duality, to algebraic identities, as well as to information science  and quantum  error correction.

\subsection{Start of our Project}
Our collaboration on pictures began in 2015 with our study of planar algebras and  statistical physics. We  showed at that time that there is a  geometric proof of the reflection positivity (RP) property for certain statistical mechanics models~\cite{JL-PPA}.  To do this we combined insights of author ZL on subfactor theory and planar algebras with insights of author AJ on parafermions, to yield a new type of planar algebra, well-suited for the description of quantum spin systems with $d$ individual states of a single particle; we call this parafermion planar para algebra (PAPPA) of order $d$, and we interpreted this parameter as a {\it charge} in the abelian group $\mathbb{Z}_{d}$.  This is a special example of a family of generalizations of planar algebras that ZL had considered. 

We found that  PAPPA gives a topological interpretation of protocols in quantum information, and in fact we providedxx a dictionary to translate in both directions between algebraic operations with Pauli matrices, Fourier transform, and measurements (on the one hand) and pictorial representations by PAPPA (on the other).  Ultimately we were able to use PAPPA to give a method of topological design for new protocols.  At that point it started to become clear to us, that central to our methods were certain mathematical languages.  The PAPPA gave a language of ``braided charged strings'' in a plane.  

We later discovered that is it useful to introduce a redundancy  by replacing  two strings with four.  Then the real physics of a single spin lies in a $d$-dimensional physically-relevant subspace of a $d^{2}$-dimensional Hilbert space that is mostly virtual.\footnote{For example, physical  subspaces in larger virtual spaces arise in many other physics contexts, of which we mention three:  { 1.} One recovers a reflection-positive subspace by pairing random fields with their reflection. { 2.} One recovers a single-photon Hilbert space of states by considering the Gupta-Bleuler analysis of defining a positive inner product on a gauge-invariant subspace of single photon configurations. { 3.} In non-abelian gauge theory one has other methods related to Becchi, Rouet, Stora, and Tyutin to single out physical states with non-negative norm.}  Implementing this correctly for multi-particle states led to the three-dimensional quon language.

One can view quon as a topological quantum field theory having certain defects, given by a special set of charged strings.  This language had immediate application both in quantum information and in algebra, and eventually for other things as well.

What is perhaps  more surprising is that one can develop a branch of Fourier analysis based on transformations of pictures.  The pictorial Fourier transform $\FS$, that we call the string Fourier transform (SFT) plays a central role in many applications of our picture languages. 
Over the course of time, we shifted focus to thinking of the languages as primary, and their application to various problems in different mathematical areas as secondary.  We have found interesting things in algebra, topology, representation theory, as well as in mathematical physics.  

\section{The Two String Language\label{Sect:TwoStrings}}
We describe how our pictures, that evolved from diagrams in planar algebra, became useful in quantum information. 
We began to study picture languages in an effort to combine our work on planar algebras with the notion of parafermions \cite{JL-PPA}.  Following a suggestion from a student, Alex Wozniakowski,  we began to study quantum information, and in particular communication  protocols.
As communication only changes the physical location of information, but not information itself, it is natural to simulate this process by topological isotopy.
In this way, we introduced a method to design multi-partite communication protocols using virtual concepts in our picture language in \cite{JLW-1,JLW-2}.
In our original language, a $d$-state quantum particle is represented by a vector (qudit)  in a $d$-dimensional Hilbert space. 

\subsection{Hilbert Space of Charged Strings}
One way to interpret Hilbert space vectors in the  theory of planar algebras is to represent the vector by a cap with two outputs. In our variation, we add a label $k\in\mathbb{Z}_{d}$, that we call charge and draw
\be\label{Basis-k}
\raisebox{-.8cm}{\scalebox{1.4 }{\tikz{
\draw (1/3,0)--(1/3,1/3)  arc (180:0:1/3) -- (1,0);
\node at (.5+1/2.7,1/3) {$\scriptstyle k$};
\node at (2/3,-.2){\size{output}};
}}}, 
\text{  where we omit the label if  $k=0$.}
\ee
The adjoint of a picture is a vertical reflection that reverses charges, turning outputs to inputs and inputs to outputs, 
\[
\scalebox{1.4}{\raisebox{-.2cm}{$\tikz{
\draw (0,0)--(0,1/3)  arc (180:0:1/3) -- (2/3,0);
\node at (1/3+1/5,1/3) {$\scriptstyle k$};
}^{*}$} =
\raisebox{-.2cm}{\tikz{$
\node at (1/3+1/15,-1/3) {$\scriptstyle-k$};
\draw (0,0)--(0,-1/3)  arc (-180:0:1/3) -- (2/3,0);
$}
}}\;.
\]
A general vector $W$ then becomes a box with two outputs,
\be\label{Vector-1}
\raisebox{-.65cm}{
\tikz{
\draw (0,2-.75) rectangle (1,2.5-.75);
\draw (.25,.75) -- (.25,1.25);
\draw (.75,.75) -- (.75,1.25);
\node at (.5,1.5) {$\scriptstyle W$};
}} \;, \phantom{pace}
\ee
 and the positivity of the inner product for a vector $W$ is written 
\[
\raisebox{-.7cm}{
\tikz{
\draw (0,2-.75) rectangle (1,2.5-.75);
\draw (0,-.5+.75) rectangle (1,.75);
\draw (.25,.75) -- (.25,1.25);
\draw (.75,.75) -- (.75,1.25);
\node at (.53,.5) {$\scriptstyle W^{*}$};
\node at (.5,1.5) {$\scriptstyle W$};
}} \,\geq0\;.
\]

A transformation $T$ has inputs and outputs.  A transformation on vectors has two of each,
\[
\scalebox{1.6}{\tikz{
\draw (-1/6,1/6) rectangle (1/6+2/3,1-1/6);
\draw (0,-1/6)--(0,1/6);
\draw (2/3,-1/6)--(2/3,1/6);
\draw (0,1+1/6)--(0,1-1/6);
\draw (2/3,1+1/6)--(2/3,1-1/6);
\node at (1/3,1/2) {{$\scriptstyle T$}};
\node at (1/3,-.3){{\size output}};
\node at (1/3,1.3){{\size input}};
}}
\raisebox{1.5cm}{\ ,}\phantom{X}
\]
that acts on a vector by vertical multiplication:
\[
\scalebox{1.2}{\tikz{
\draw (0,1/6+1)--(0,2/6+1)  arc (180:0:1/3) -- (2/3,1/6+1);
\node at (1/3+1/6,15/12) {$\scriptstyle k$};
\draw (-1/6,1/6) rectangle (1/6+2/3,1-1/6);
\draw (0,-1/6)--(0,1/6);
\draw (2/3,-1/6)--(2/3,1/6);
\draw (0,1+1/6)--(0,1-1/6);
\draw (2/3,1+1/6)--(2/3,1-1/6);
\node at (1/3,1/2)  {{${\scriptstyle T}$}};
\node at (1/3,-.3){{\size output}}
}}
\raisebox{1cm}{\ .}\phantom{X}
\]
The circle graph is called the quantum dimension,
\[
\raisebox{-.375cm}{\tikz{
\draw (0,0) circle [radius=0.5];
\node at (-1,0) {$\delta =$};
}}
\;,
\quad \text{and for $k\neq0$,  \ }
\raisebox{-.375cm}{\tikz{
\draw (0,0) circle [radius=0.5];
\node at (-1,0) {$0 =$};
\node at (.3,0) {$k$};
}}
\;.
\]
In order to ensure many of the following relations, it is necessary to choose the quantum dimension $\delta=\sqrt{d}$. This is the smallest value of $\delta$ consistent with the trace of an arbitrarily large system being positive, so it yields a Hilbert space inner product on pictures~\cite{JonesPA,JL-PPA}.

\subsection{Relations}
Assuming multiplication of charges on a line is additive, the basis is orthogonal and the identity transformation $I$ on qudit  yields the resolution of the identity,  $I= \sum_{k\in\mathbb{Z}_{d}} P_{k}$.    This is given in terms of the projections $P_{k}$ onto the vectors \eqref{Basis-k}, as
\be\label{Projection-k}
 P_{k}= \frac{1}{\delta} 
\raisebox{-.5cm}{
\tikz{
\fqudit {0}{-3.5\nn}{1\mn}{.5\nn}{\phantom{kk} k}
\fmeasure {0}{0}{1\mn}{.5\nn}{\ -k}
}}\;.
\ee
In terms of pictures this is the relation,
\be\label{Identity-Relation}
\raisebox{-.36cm}{
\tikz{
\draw (0,0)--(0,1);
\draw (2/3,0)--(2/3,1);
}}
=  \frac{1}{\delta} \sum_{k\in\mathbb{Z}_{d}}
\raisebox{-.5cm}{
\tikz{
\fqudit {0}{-3.5\nn}{1\mn}{.5\nn}{\phantom{kk} k}
\fmeasure {0}{0}{1\mn}{.5\nn}{\ -k}
}}
\;.
\ee

Let $q=e^{2\pi i /d}$.  We assume vertical motion of the charges to reverse the order of multiplication 
\be\label{PI-Relation}
\raisebox{-.4cm}{
\begin{tikzpicture}
\draw (0,0) --(0,01);
\node (0,0) at (-0.15,0.2) {$k$};
\draw (1.15,0) --(1.15,1);
\node (1.2,0) at (1,0.8) {$\ell$};
\end{tikzpicture}
}
=q^{k\ell}
\raisebox{-.4cm}{
\begin{tikzpicture}
\draw (0,0) --(0,01);
\node (0,0) at (-0.15,0.8) {$k$};
\draw (1.15,0) --(1.15,1);
\node (1.2,0) at (1,0.2) {$\ell$};
\end{tikzpicture}}\;.
\ee
We call this a ``para isotopy'' relation.
It is natural to consider the para-isotopy interpolation for the charges at the same level, and to define 
\be\label{Interpolation-Relation}
\raisebox{-.4cm}{
\begin{tikzpicture}
\draw (0,0) --(0,01);
\node (0,0) at (-0.15,0.2) {$k$};
\draw (1.15,0) --(1.15,1);
\node (1.2,0) at (.85,0.8) {$-k$};
\end{tikzpicture}
}
=\zeta^{-k^{2}}
\raisebox{-.4cm}{
\begin{tikzpicture}
\draw (0,0) --(0,01);
\node (0,0) at (-0.15,0.5) {$k$};
\draw (1.15,0) --(1.15,1);
\node (1.2,0) at (.85,0.5) {$-k$};
\end{tikzpicture}}\;,
\ee
where 
\be\label{zeta}
\zeta=q^{1/2}\;,
\quad\text{and}\quad
\zeta^{d^{2}}=1\;,
\ee
ensuring $\zeta^{k^{2}}=\zeta^{(k+d) ^{2}}$.  We also use the following relation to move a charge across a string with fixed endpoints:
\be\label{SF-Relation}
   \raisebox{-0.2cm}{
   \tikz{
   \fqudit {0}{0}{1\mn}{1\nn}{}
   \node at (-3\mn,1\nn) {{$\hskip -2.5cm k$}};
   }}
    =\zeta^{k^2}~~
    \raisebox{-0.2cm}{\tikz{\fqudit {0}{0}{1\mn}{1\nn}{}
       \node at (-1.4\mn,1\nn) {{$k$}};
    }}\;,
 \quad\text{so}\quad
   \raisebox{-0.2cm}{
   \tikz{
   \fmeasure {0}{0}{1\mn}{1\nn}{}
   \node at (-1.3\mn,-.8\nn) {$\scriptstyle -k$};
   }}
    =\zeta^{k^2}~~
    \raisebox{-0.2cm}{\tikz{\fmeasure{0}{0}{1\mn}{1\nn}{}
   \node at (.7\mn,-.8\nn) {$\scriptstyle -k$};
       \node at (-1.6\mn,1\nn) {};
    }}\;.
\ee

\subsection{Basic Transformations}
The basic transformations in quantum information protocols involve the identity, the Pauli matrices $X,Y,Z$, the Fourier transform (or Hademard) matrix $F$, and the phase (Gaussian) transformation $G$ that implements a braid.  Pictorially, 
\be
I=
\raisebox{-.4cm}{
\tikz{
\draw (0,0)--(0,1);
\draw (2/3,0)--(2/3,1);
}}\;,\qquad
X=
\raisebox{-.4cm}{
\tikz{
\draw (0,0)--(0,1);
\draw (2/3,0)--(2/3,1);
\node at (1.5/3,1/2) {\scriptsize 1};
}}\;,\qquad
Y=
\raisebox{-.4cm}{
\tikz{
\draw (0,0)--(0,1);
\draw (2/3,0)--(2/3,1);
\node at (-.5/3,1/2) {\scriptsize -1};
}}\;,\qquad
Z=
\raisebox{-.4cm}{
\tikz{
\draw (0,0)--(0,1);
\draw (2/3,0)--(2/3,1);
\node at (-.5/3,1/2) {\scriptsize 1};
\node at (1.5/3,1/2) {\scriptsize -1};
}}\;.
\ee

\subsection{The Analytic Fourier Transform}
Consider the projection $P_{k}$ onto the $k^{\rm th}$ basis vector~\eqref{Projection-k}.
The ordinary Fourier transform $F$ acting on this picture $P_{k}$ as a function of $k$ is
\[
FP_{k} = \frac{1}{\sqrt{d}} \sum_{\ell\in\mathbb{Z}_{d}} q^{k\ell} P_{\ell}\;,
\]
or 
\be\label{Fourier}
F 
\raisebox{-.5cm}{
\tikz{
\fqudit {0}{-3.5\nn}{1\mn}{.5\nn}{\phantom{kk} k}
\fmeasure {0}{0}{1\mn}{.5\nn}{\ -k}
}}
= \frac{1}{\sqrt{d}}\sum_{\ell\in\mathbb{Z}_{d}} q^{k\ell}
\raisebox{-.5cm}{
\tikz{
\fqudit {0}{-3.5\nn}{1\mn}{.5\nn}{\phantom{kk} \ell}
\fmeasure {0}{0}{1\mn}{.5\nn}{\ -\ell}
}}\;.
\ee

\subsection{The pictorial Fourier transform}
The pictorial Fourier transform $\FS$, or string Fourier transform (SFT),  rotates the one-qudit transformation  picture by $90^{\circ}$. In other words, 
\[
\raisebox{-.5cm}{
\tikz{
\node at (-1/3 + -1/6,1/2) {{$\mathfrak{F}_{s}$}};
\draw (-1/6,1/6) rectangle (1/6+2/3,1-1/6);
\draw (0,-1/6)--(0,1/6);
\draw (2/3,-1/6)--(2/3,1/6);
\draw (0,1+1/6)--(0,1-1/6);
\draw (2/3,1+1/6)--(2/3,1-1/6);
\node at (1/3,1/2) {\size{$T$}};
}}
=
\raisebox{-.5cm}{
\tikz{
\draw (-1/6,1/6) rectangle (1/6+2/3,1-1/6);
\draw (0,-1/6)--(0,1/6);
\draw (2/3,-1/6)--(2/3,1/6);
\draw (0,1+1/6)--(0,1-1/6);
\draw (2/3,1+1/6)--(2/3,1-1/6);
\node at (1/3,1/2) {\size{$\mathfrak{F}_{s} T$}};
}}
=
\raisebox{-.6cm}{\scalebox{.7}[.8]{
\tikz{
\draw (-1/6,1/6) rectangle (1/6+2/3,1-1/6);
\draw (0,-1/6) arc (0:-180:.25);
\draw (2/3,1+1/6) arc (180:0:.25);
\draw (-.5,-1/6)--(-.5,8/6);
\draw (2/3+.5,-2/6)--(2/3+.5,1+1/6);
\draw (0,-1/6)--(0,1/6);
\draw (2/3,-2/6)--(2/3,1/6);
\draw (0,1+ 2/6)--(0,1-1/6);
\draw (2/3,1+1/6)--(2/3,1-1/6);
\node at (1/3,1/2) {\size{$T$}};
}}}\;.
\]
This  transform $\FS$ applied to a neutral transformatfon $T$ generalizes the Fourier transform $F$. 
\begin{theorem} \label{Thm:SFT-F} On diagonal $d\times d$ matrices, the transformations $\FS$ and $F$ agree $\FS=F$.  The picture theorem is, compare with \eqref{Fourier},
\be
\fbox{$\FS 
\raisebox{-.5cm}{
\tikz{
\fqudit {0}{-3.5\nn}{1\mn}{.5\nn}{\phantom{kk} k}
\fmeasure {0}{0}{1\mn}{.5\nn}{\ -k}
}}
=
\raisebox{-.4cm}{
\begin{tikzpicture}
\draw (0,0) --(0,01);
\node  at (-0.15,0.5) {$\scriptstyle k$};
\draw (.8,0) --(.8,1);
\node  at (.55,0.5) {$\scriptstyle -k$};
\end{tikzpicture}}
=F
\raisebox{-.5cm}{
\tikz{
\fqudit {0}{-3.5\nn}{1\mn}{.5\nn}{\phantom{kk} k}
\fmeasure {0}{0}{1\mn}{.5\nn}{\ -k}
}}
$}
\;.
\ee
\end{theorem}

\begin{proof}  The proof takes only a couple of lines.  We use the definition of $\FS$ and the properties \eqref{Projection-k} and \eqref{Fourier} to write
\beqs
\FS 
\raisebox{-.5cm}{
\tikz{
\fqudit {0}{-3.5\nn}{1\mn}{.5\nn}{}
\fmeasure {0}{0}{1\mn}{.5\nn}{}
\node at (.45,-.2) {$\scriptstyle -k$};
\node at (.55,-.95) {$\scriptstyle k$};
}}
&=&\zeta^{k^{2}}\, \FS \hskip -.25cm
\raisebox{-.5cm}{
\tikz{
\fqudit {0}{-3.5\nn}{1\mn}{.5\nn}{}
\fmeasure {0}{0}{1\mn}{.5\nn}{}
\node at (-.25,-.2) {$\scriptstyle -k$};
\node at (.55,-.95) {$\scriptstyle k$};
}}
= \zeta^{k^{2}}
\raisebox{-.44cm}{
\tikz{
\node at (-.15,.2)  {$\scriptstyle k$};
\node at (.55,.8)  {$\scriptstyle -k$};
\draw (0,0)--(0,1);
\draw (.8,0)--(.8,1);
}}
= 
\frac{1}{\delta}\sum_{\ell\in\mathbb{Z}_{d}}  \zeta^{k^{2}} 
\raisebox{-.6cm}{
\tikz{
\fqudit {0}{-3.5\nn}{1\mn}{.5\nn}{}
\fmeasure {0}{0}{1\mn}{.5\nn}{}
\node at (-.1,-1) {$\scriptstyle k$};
\node at (.5,-.9) {$\scriptstyle \ell$};
\node at (.4,-.1)  {$\scriptstyle -k$};
\node at (.4,-.3)  {$\scriptstyle -\ell$};
}}\\
&=& 
\frac{1}{\delta}\sum_{\ell\in\mathbb{Z}_{d}}   \zeta^{k^{2}} q^{k\ell}
\raisebox{-.6cm}{
\tikz{
\fqudit {0}{-3.5\nn}{1\mn}{.5\nn}{}
\fmeasure {0}{0}{1\mn}{.5\nn}{}
\node at (-.1,-.9) {$\scriptstyle k$};
\node at (.5,-1) {$\scriptstyle \ell$};
\node at (.4,-.1)  {$\scriptstyle -k$};
\node at (.4,-.3)  {$\scriptstyle -\ell$};
}}
= \frac{1}{\delta}\sum_{\ell\in\mathbb{Z}_{d}} q^{k^{2}+k\ell}
\raisebox{-.6cm}{
\tikz{
\fqudit {0}{-3.5\nn}{1\mn}{.5\nn}{}
\fmeasure {0}{0}{1\mn}{.5\nn}{}
\node at (.5,-1.1) {$\scriptstyle k$};
\node at (.5,-.9) {$\scriptstyle \ell$};
\node at (.4,-.1)  {$\scriptstyle -k$};
\node at (.4,-.3)  {$\scriptstyle -\ell$};
}}\\
&=& \frac{1}{\delta}\sum_{\ell\in\mathbb{Z}_{d}} q^{k\ell}
\raisebox{-.6cm}{
\tikz{
\fqudit {0}{-3.5\nn}{1\mn}{.5\nn}{}
\fmeasure {0}{0}{1\mn}{.5\nn}{}
\node at (.5,-.95) {$\scriptstyle \ell$};
\node at (.4,-.2)  {$\scriptstyle -\ell$};
}}
=\frac{\sqrt{d}}{\delta}\,F
\raisebox{-.5cm}{
\tikz{
\fqudit {0}{-3.5\nn}{1\mn}{.5\nn}{\phantom{kk} k}
\fmeasure {0}{0}{1\mn}{.5\nn}{\ -k}
}}
=F
\raisebox{-.5cm}{
\tikz{
\fqudit {0}{-3.5\nn}{1\mn}{.5\nn}{\phantom{kk} k}
\fmeasure {0}{0}{1\mn}{.5\nn}{\ -k}
}}\;.
\eeqs
In the first equality we use the relation \eqref{SF-Relation}, while the second is the definition of the SFT.  In the third identity we insert the resolution of the identity relation \eqref{Identity-Relation}.  
Next we use the para-isotopy relation \eqref{PI-Relation}, and after that we again apply \eqref{SF-Relation}.
In the next-to-last equality we translate the sum,  replacing $\ell\in\mathbb{Z}_{d}$ with $\ell-k$. Thus $\FS=F$ on these projections.  Use the definition of para-isotopy interpolation to identify the middle term in the statement of the picture theorem, as the third expression in the calculation above. 
\end{proof}

\section{The Reflection-Positivity Property}
The reflection-positivity property has had wide applications in mathematical physics (quantum field theory, statistical physics). It is central in the mathematical study of phase transitions and also  the relations between the two subjects. It is also important when one considers abstract Fourier analysis.

 In this section we use a single  string (tensor network) notation for a vector $W\in\Hil$, representing  it by the picture \raisebox{-.35cm}{\scalebox{.6}{$\tikz{
\draw (0,2) rectangle (1,2.5);
\draw (.5,1.5) -- (.5,2);
\node at (.5,2.24) {$\scriptstyle W$};
}$}}\;, in place of \eqref{Vector-1}.

\subsection{Reflection}
In the definition of the string Fourier transform $\FS$, we considered  a finite-dimensional Hilbert space $\Hil$ and its dual $\Hild=\theta \Hil$.  The map $\theta$ denotes the Riesz representation map  $\Hil\to\Hild$. 
Then duality $\theta$ and the adjoint $^{*}$ are related; 
\be\label{TurnAround}
\scalebox{.7}{
\raisebox{-.65cm}{
\tikz{
\draw (0,2) rectangle (1,2.5);
\draw (.5,1.5) -- (.5,2);
\node at (.5,2.24) {$\scriptstyle \theta(W)$};
}}\ =
\raisebox{-.65cm}{
\tikz{
\draw (0,2) rectangle (1,2.5);
\draw (.5,2.5) -- (.5,3);
\node at (.5,2.24) {$\scriptstyle W^{*}$};
\fqudit{.5}{3}{-.5}{0}{}
\draw (1.5,3) -- (1.5,1.5);
}}}\;,
\quad\text{\small and}\quad
\scalebox{.7}{
\raisebox{-.2cm}{
\tikz{
\draw (0,2) rectangle (1,2.5);
\draw (.5,2.5) -- (.5,3);
\node at (.5,2.24) {$\scriptstyle \theta(W)^{*}$};
}}\ =
\raisebox{-1.15cm}{
\tikz{
\draw (0,2) rectangle (1,2.5);
\draw (.5,1.5) -- (.5,2);
\node at (.5,2.24) {$\scriptstyle W$};
\fmeasure{.5}{1.5}{-.5}{0}{}
\draw (1.5,3) -- (1.5,1.5);
}}}\;.
\ee

\subsection{Reflection Positivity}
We formulate the reflection-positivity property that we considered in~\cite{JL-PPA} in a recent work~\cite{JL-LW} in terms of $\FS$ acting as  a  map
	\[
	\FS:\hom(\Hild\otimes \Hil) \to \hom(\Hil\otimes \Hild)\;.
	\] 
\begin{definition}
A map $H\in \hom(\Hild\otimes \Hil)$ has the RP property if for every $V, W\in \Hil$, 
	\[
	\lra{\theta(V)\otimes V, e^{-\beta H}\theta(W)\otimes W}_{\Hild\otimes \Hil} \geq0\;.
	\]
\noindent Pictorially, the RP property means that  for all $V,W\in\Hil$,
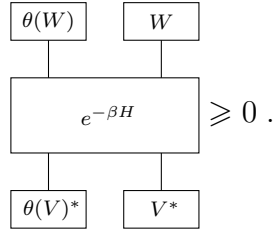
\[
\raisebox{-1.45cm}{
\tikz{
\draw (0,2) rectangle (1,2.5);
\draw (0,-.5) rectangle (1,0);
\draw (0,.5) rectangle (2.5,1.5);
\draw (.5,0) -- (.5,.5);
\draw (.5,1.5) -- (.5,2);
\draw (1.5,2) rectangle (2.5,2.5);
\draw (1.5,-.5) rectangle (2.5,0);
\draw (2,0) -- (2,.5);
\draw (2,1.5) -- (2,2);
\node at (1.3,1) {$\scriptstyle e^{-\beta H}$};
\node at (.55,.-.26) {$\scriptstyle \theta(V)^{*}$};
\node at (2.05,.-.24) {$\scriptstyle V^{*}$};
\node at (.5,2.24) {$\scriptstyle\theta(W)$};
\node at (2,2.25) {$\scriptstyle W$};
}}\geq0\;.
\]
\end{definition}
Here vertical reflection symmetry of the picture is not apparent.  However there appears to be some horizontal symmetry.  We relate the Fourier transform of the  terms in the power series of  $e^{-\beta H}$ to the convolution powers of the Fourier transform on $-H$.  

\begin{remark}
Considering the tensor product $\Hild\otimes \Hil$ means we consider here the special ``bosonic'' case. The construction can be modified to cover possible twisting. We now outline the pictorial proof of the following result.  
\end{remark}

\begin{theorem}
If $\FS(-H)\geq0$, then $\FS(e^{-\beta H})\geq0$, and  $H$ has the RP property.
\end{theorem}
\begin{proof}
Pictorially the assumption that  $\FS(-H)\geq0$, means that 
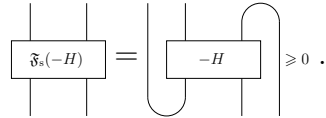
\[
\scalebox{.5}{\raisebox{-1.4cm}{$
\tikz{
\draw (0,.5) rectangle (2.5,1.5);
\draw (.5,-.5) -- (.5,.5);
\draw (.5,1.5) -- (.5,2.5);
\draw (2,-.5) -- (2,.5);
\draw (2,1.5) -- (2,2.5);
\node at (1.2,1) {$\FS (- H)$};
}$}}
=
\scalebox{.5}{\raisebox{-1.4cm}{$
\tikz{
\draw (0,.5) rectangle (2.5,1.5);
\draw (.5,0) -- (.5,.5);
\draw (.5,1.5) -- (.5,2.5);
\fmeasure{-.5}{0}{-.5}{0}{};
\draw (-.5,0) -- (-.5,2.5);
\fqudit{2}{2}{-.5}{0}{}
\draw (2,-.5) -- (2,.5);
\draw (2,1.5) -- (2,2);
\draw (3,2) -- (3,-.5);
\node at (1.2,1) {$- H$};
}$}$\ \geq0$}\;.
\]
We  prove in reference \cite{JL-LW} that positivity of $\FS (-H)$ entails positivity of $\FS(e^{-\beta H})$ for $\beta\geq0$.  Thus, 
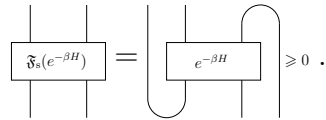
\[
\scalebox{.5}{\raisebox{-1.4cm}{$
\tikz{
\draw (0,.5) rectangle (2.5,1.5);
\draw (.5,-.5) -- (.5,.5);
\draw (.5,1.5) -- (.5,2.5);
\draw (2,-.5) -- (2,.5);
\draw (2,1.5) -- (2,2.5);
\node at (1.2,1) {$\FS(e^{-\beta H})$};
}$}}
=
\scalebox{.5}{\raisebox{-1.4cm}{$
\tikz{
\draw (0,.5) rectangle (2.5,1.5);
\draw (.5,0) -- (.5,.5);
\draw (.5,1.5) -- (.5,2.5);
\fmeasure{-.5}{0}{-.5}{0}{};
\draw (-.5,0) -- (-.5,2.5);
\fqudit{2}{2}{-.5}{0}{}
\draw (2,-.5) -- (2,.5);
\draw (2,1.5) -- (2,2);
\draw (3,2) -- (3,-.5);
\node at (1.2,1) {$e^{-\beta H}$};
}$}$\ \geq0$}\;.
\]
The picture proof of the RP property is evident, since $\Hil$ is a Hilbert space. Using \eqref{TurnAround}, one sees that 
\be\label{RPbyPics}
\raisebox{-1.45cm}{
\tikz{
\draw (0,2) rectangle (1,2.5);
\draw (0,-.5) rectangle (1,0);
\draw (0,.5) rectangle (2.5,1.5);
\draw (.5,0) -- (.5,.5);
\draw (.5,1.5) -- (.5,2);
\draw (1.5,2) rectangle (2.5,2.5);
\draw (1.5,-.5) rectangle (2.5,0);
\draw (2,0) -- (2,.5);
\draw (2,1.5) -- (2,2);
\node at (1.3,1) {$\scriptstyle e^{-\beta H}$};
\node at (.55,.-.25) {$\scriptstyle \theta(V)^{*}$};
\node at (2.05,.-.25) {$\scriptstyle V^{*}$};
\node at (.5,3-.75) {$\scriptstyle\theta(W)$};
\node at (2,3-. 75) {$\scriptstyle W$};
}}=
\raisebox{-1.45cm}{
\tikz{
\draw (0,2) rectangle (1,2.5);
\draw (0,-.5) rectangle (1,0);
\draw (0,.5) rectangle (2.5,1.5);
\draw (.5,0) -- (.5,.5);
\draw (.5,1.5) -- (.5,2);
\draw (1.5,2) rectangle (2.5,2.5);
\draw (1.5,-.5) rectangle (2.5,0);
\draw (2,0) -- (2,.5);
\draw (2,1.5) -- (2,2);
\node at (1.3,1) {$\scriptstyle \FS(e^{-\beta H})$};
\node at (.55,.-.25) {$\scriptstyle V^{*}$};
\node at (2.05,.-.25) {$\scriptstyle \theta(W)^{*}$};
\node at (.5,3-.75) {$\scriptstyle V$};
\node at (2,3-. 75) {$\scriptstyle \theta(W)$};
}}
\geq0\;.
\ee
To justify the positivity statement, look at the picture on the right of the equal sign. The  vector on the bottom is the vertical reflection of the vector on top.  Algebraically, the relation \eqref{RPbyPics} between pictures means 
\beqs
&&\hskip-.5in \lra{\theta(V)\otimes V, e^{-\beta H}(\theta(W)\otimes W)}_{\Hild\otimes \Hil}
\\
&&\qquad =
\lra{V \otimes \theta(W), \FS(e^{-\beta H})(V\otimes \theta(W))}_{\Hil\otimes \Hild}\;.
\eeqs
Thus $\FS(e^{-\beta H})\geq0$ for all $\beta\geq0$ guarantees the RP property, though this is a stronger condition than required.
\end{proof}

One can carry the picture proof further.  
Extend $\theta$ to $\Hild\otimes \Hil$, so that $\theta (V\otimes W)=\theta(W)\otimes \theta(V)$.  Let $\theta(H)=\theta H \theta$.
\begin{theorem}
In case $H=H_{-}+H_{0}+H_{+}$ where $H_{+}=\theta(H_{-})\in\hom(\Hil)$, and $\FS(-H_{0})\geq0$, then $H$ has the RP property.
\end{theorem}

We also have used a generalization of this method to establish the RP property in Levin-Wen models~\cite{JL-LW}, now a fashionable model to describe topological insulators.

\section{The Quon Language and Virtual States}
The three dimensional quon language can work even better than the two-dimensional pictures discussed above, and this brings to the fore the concept of virtual states. The 2D picture language could be considered as the local boundary theory of the 3D quon language. For a single-particle Hilbert space we take two copies of the two-string language in \S\ref{Sect:TwoStrings}. Then one can represent the Pauli matrices for a single spin by neutral transformations:
\be\label{4Pauli}
\scalebox{.7}{
I=
\raisebox{-.4cm}{
\tikz{
\draw (0,0)--(0,1);
\draw (2/3,0)--(2/3,1);
\draw (4/3,0)--(4/3,1);
\draw (2,0)--(2,1);
}}\;,\qquad
X=
\raisebox{-.4cm}{
\tikz{
\draw (0,0)--(0,1);
\draw (2/3,0)--(2/3,1);
\draw (4/3,0)--(4/3,1);
\draw (2,0)--(2,1);
\node at (-.5/3,1/2) {\scriptsize 1};
\node at (1.8,1/2) {\scriptsize -1};
}}\;,\qquad
Y=
\raisebox{-.4cm}{
\tikz{
\draw (0,0)--(0,1);
\draw (2/3,0)--(2/3,1);
\draw (4/3,0)--(4/3,1);
\draw (2,0)--(2,1);
\node at (-.5/3,1/2) {\scriptsize -1};
\node at (3.5/3,1/2) {\scriptsize 1};
}}\;,\qquad
Z=
\raisebox{-.4cm}{
\tikz{
\draw (0,0)--(0,1);
\draw (2/3,0)--(2/3,1);
\draw (4/3,0)--(4/3,1);
\draw (2,0)--(2,1);
\node at (-.5/3,1/2) {\scriptsize 1};
\node at (1.5/3,1/2) {\scriptsize -1};
}}}\;,
\ee
where neutral means the total change is zero in $\mathbb{Z}_d$.
Each of these transformations is unitary; each equals the identity when raised to the power~$d$. 
Also each of these Pauli matrices  is neutral, and each commutes with  the  grading transformation 
\be
\scalebox{.7}{
$\gamma$ = 
\raisebox{-.4cm}{
\tikz{
\draw (0,0)--(0,1);
\draw (2/3,0)--(2/3,1);
\draw (4/3,0)--(4/3,1);
\draw (2,0)--(2,1);
\node at (-.5/3,1/2) {\scriptsize 1};
\node at (1.5/3,1/2) {\scriptsize -1};
\node at (.5+2/3,1/2) {\scriptsize 1};
\node at (.5+4/3,1/2) {\scriptsize -1};
}}}=\zeta^{-1}  XYZ\;,
\ee
where $\zeta$ is given in \eqref{zeta}.
Therefore each Pauli matrix acts on an eigenspace of $\gamma$.  On the neutral eigenspace for which  $\gamma=1$, these matrices satisfy the desired relation $XYZ=\zeta$.  We call the $\gamma=1$ space, the $d$-dimensional \emph{physical} or \emph{real} subspace of the $d^{2}$-dimensional vector space of eigenstates of $Z$. 

This description of a single particle excitation as a neutral pair led to our realization that particles are best described by pictures in one higher dimension. A three-particle state is given by a cap enclosing the four output lines with internal connections denoting vectors $v_{1}, v_{2}, v_{3}$,
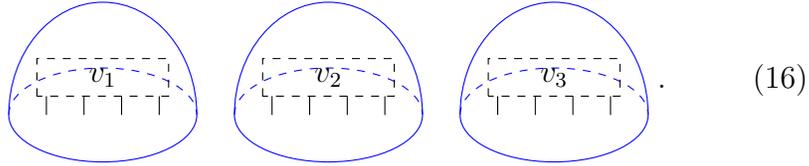
\be\label{3Particles}
\raisebox{-1.cm}{
\begin{tikzpicture}
\begin{scope}[xscale=.5,yscale=.25]
\foreach \x in {0,6,12}{
\draw[blue] (\x,12) arc (-180:0:2.5);
\draw[blue] (\x,12) arc (180:0:2.5 and 6);
\draw[blue,dashed] (\x,12) arc (180:0:2.5);
\draw[dashed] (\x+.75,13) rectangle (\x+4.25,15);
\foreach \y in {1,2,3,4}
{
\draw (\x+\y,12)--++(0,1);
}
}
\node at (0+2.5,14) {$v_{1}$};
\node at (6+2.5,14) {$v_{2}$};
\node at (12+2.5,14) {$v_{3}$};
\end{scope}
\end{tikzpicture}}\;.
\ee
A transformation can be represented by a 3D-TQFT, 
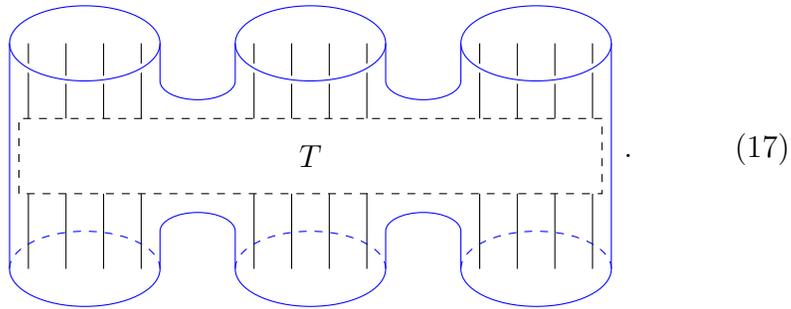
\be
\raisebox{-2cm}{
\begin{tikzpicture}
\begin{scope}[xscale=.5,yscale=.25]
\draw[dashed] (.75,4) rectangle (16.25,8);
\node at (8.5,6) {$T$};
\foreach \x in {1,2,3,4,7,8,9,10,13,14,15,16} {
\draw (\x,0)--++(0,4);
\draw (\x,8)--++(0,4);
}

\foreach \x in {0.5,6.5,12.5}{
\draw[white,WL]  (\x,12) arc (-180:180:2);
\draw[blue] (\x,12) arc (-180:180:2);
\draw[blue] (\x,0) arc (-180:0:2);
\draw[blue,dashed] (\x,0) arc (180:0:2);
}
\draw[blue] (0.5,0)--++(0,12);
\draw[blue] (16.5,0)--++(0,12);
\draw[blue] (4.5,12)--++(0,-2) arc (-180:0:1)--++(0,2);
\draw[blue] (10.5,12)--++(0,-2) arc (-180:0:1)--++(0,2);
\draw[blue] (4.5,0)--++(0,2) arc (180:0:1)--++(0,-2);
\draw[blue] (10.5,0)--++(0,2) arc (180:0:1)--++(0,-2);
\end{scope}
\end{tikzpicture}}\;.
\ee
In fact it is more symmetric to arrange the output lines on the corners of a square, with an arrow indicating the numbering of the strings.  For clarity we suppress the three-dimensional TQFT, except for the outline (here a circle) of a cross-section. Looking down on the three-particle outputs, and suppressing the output lines, one can represent the three-particle state in \eqref{3Particles} as
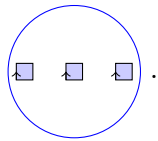
\[
\raisebox{-.8cm}{\begin{tikzpicture}
\begin{scope}[scale=.22]
\foreach \x in {0,1}{
\foreach \y in {0,1}{
\foreach \u in {0}{
\foreach \v in {1,2,3}{
\coordinate (A\u\v\x\y) at (\x+1.5*\u,\y+3*\v);
}}}}
\begin{scope}[shift={(24,6)},rotate=90]
\draw[blue] (.5,6.5) circle (4);
\foreach \x in {0,1}{
\foreach \y in {0,1}{
\foreach \u in {0}{
\foreach \v in {1,2,3}{
\coordinate (A\u\v\x\y) at (\x+1.5*\u,\y+3*\v);
}}}}
\foreach \u in {0}{
\foreach \v in {1,2,3}{
\fill[blue!20] (A\u\v00) rectangle (A\u\v11);
\draw (A\u\v00) rectangle (A\u\v11);
\draw[->] (A\u\v01)--++(.5,0);
}}
\end{scope}
\end{scope}
\end{tikzpicture}}\;.
\]

By defining the internal connections in two different ways, one obtains the three-particle Greenberger-Horne-Zeilinger resource state $|\text{GHZ}\rangle_{3}$~\cite{GHZ}, or its SFT the state $|{\text{Max}}\rangle_{3}$~\cite{JLW-1}, as 
\begin{equation}\label{GHZ to MAX}
\scalebox{1}{
\raisebox{-.8cm}{
\begin{tikzpicture}
\node at (-1.85,1.5) {$|\text{GHZ}\rangle_{3}=$};
\node at (3.15,1.5) {$|\text{Max}\rangle_{3}=$};
\begin{scope}[scale=.22]
\draw[blue] (.5,6.5) circle (4);
\foreach \x in {0,1}{
\foreach \y in {0,1}{
\foreach \u in {0}{
\foreach \v in {1,2,3}{
\coordinate (A\u\v\x\y) at (\x+1.5*\u,\y+3*\v);
}}}}

\foreach \u in {0}{
\foreach \v in {1,2,3}{
\fill[blue!20] (A\u\v00) rectangle (A\u\v11);
\draw (A\u\v00) rectangle (A\u\v11);
\draw[->] (A\u\v00)--++(0,.5);
}}

\draw (A0101) to [bend left=30] (A0200);
\draw (A0201) to [bend left=30] (A0300);
\draw (A0301) to [bend left=-30] (A0100);

\draw (A0111) to [bend left=-30] (A0210);
\draw (A0211) to [bend left=-30] (A0310);
\draw (A0311) to [bend left=30] (A0110);

\begin{scope}[shift={(30,6)},rotate=90]
\draw[blue] (.5,6.5) circle (4);
\foreach \x in {0,1}{
\foreach \y in {0,1}{
\foreach \u in {0}{
\foreach \v in {1,2,3}{
\coordinate (A\u\v\x\y) at (\x+1.5*\u,\y+3*\v);
}}}}

\foreach \u in {0}{
\foreach \v in {1,2,3}{
\fill[blue!20] (A\u\v00) rectangle (A\u\v11);
\draw (A\u\v00) rectangle (A\u\v11);
\draw[->] (A\u\v01)--++(.5,0);
}}

\draw (A0101) to [bend left=30] (A0200);
\draw (A0201) to [bend left=30] (A0300);
\draw (A0301) to [bend left=-30] (A0100);

\draw (A0111) to [bend left=-30] (A0210);
\draw (A0211) to [bend left=-30] (A0310);
\draw (A0311) to [bend left=30] (A0110);
\end{scope}
\end{scope}
\end{tikzpicture}}
}\;.
\end{equation}
This language illustrates these two states as related by rotation, given by the pictorial Fourier transform~$\FS$. For further details, see~\cite{JLW-1,JLW-Quon}. Both GHZ and Max are  maximally entangled resource vector states~\cite{GHZ,JLW-1}, so the string Fourier transform produces maximal entanglement in one operation.

ZL applied the quon language and its virtual concepts to design graphical quantum error-correcting codes in \cite{Liu-QECC}, whose logical qubits can be implemented as the ground states of translation-invariant, gapped, local Hamiltonians on quasi two-dimensional exactly solvable models. 
These codes can be regarded as surface codes with ``point defects.''

\section{Quantum Fourier Analysis}
One  theme that runs through much of the work in the picture language program comes from analytic properties of the pictorial Fourier transform~$\FS$. 
ZL with Chunlan Jiang and Jinsong Wu began to analyze analytic properties of the SFT. They found that standard inequalities such as the Young and Hausdorff-Young estimates, as well as classical uncertainly principles such as those of Heisenberg, Hirschmann or Beckner, have generalizations in which the analytic Fourier transform $F$ is replaced by the pictorial Fourier transform $\FS$~\cite{JiLWu-1, JiLWu-4, JiLWu-3}.  The study of the analytic properties of $\FS$ has become known as {\em quantum Fourier analysis}, see \cite{QFA} for an overview. This is a new subject, with only a few of many potentially interesting things already discovered; it has become a central part of our picture language project.

The pictorial proof of inequalities allowed them to determine the extremizers for the inequalities, namely to solve the inverse problem. Now these methods have been extended to analysis on quantum groups and other frameworks, and this leads to an entire realm of analytical questions that should be investigated.  

For example, one has shown that with the trace norm, defined by pairing input and output strings in a non-overlapping fashion,
\be
\text{Tr}(A^{*}B)=
\raisebox{-1.3cm}{\scalebox{.6}{
\begin{tikzpicture}
\draw (1.8,-0.3) --(1.8,0);
\coordinate (A) at (1.8,-0.3);
\coordinate (B) at (2,-0.6);
\draw[color=black] (A) to [bend left=-15] (B);
\draw (2.4,-0.2) --(2.4,0);
\coordinate (A) at (2.4,-0.2);
\coordinate (B) at (2.6,-0.5);
\draw[color=black] (A) to [bend left=-20] (B);
\node (0,0.65) at (2.1,-0.2) {$\cdots$};
\draw (1.7,0) rectangle +(0.8,0.6);
\node (0,0.65) at (2.15,0.25) {$A^{\ast}$};
\draw (1.8,0.6) --(1.8,0.9);
\draw (2.4,0.6) --(2.4,0.9);
\node (0,0.65) at (2.1,0.7) {$\cdots$};
\draw (1.7,0.9) rectangle +(0.8,0.6);
\node (0,0.65) at (2.1,1.2) {$B$};
\draw (1.8,1.5) --(1.8,1.8);
\coordinate (A) at (1.8,1.8);
\coordinate (B) at (2,2.35);
\draw[color=black] (A) to [bend left=15] (B);
\draw (2.4,1.5) --(2.4,1.7);
\coordinate (A) at (2.4,1.7);
\coordinate (B) at (2.6,2.06);
\draw[color=black] (A) to [bend left=20] (B);
\node (0,0.65) at (2.1,1.7) {$\cdots$};
\node (0,0.65) at (5.5,0.7) {$\cdots$};
\draw (2.6,-0.5) arc (-127:127:1.6);
\draw (2,-0.6) arc (-138:138:2.2);
\end{tikzpicture}
}}\;.
\ee
Assuming that the trace is positive, the GNS representation allows one to define the corresponding norms $\Vert A \Vert_{p} =\lrp{ \text{Tr}((A^{*}A)^{p/2})}^{1/p}$.  One finds inequalities that generalize inequalities of classical Fourier analysis.  This is what we call {\em quantum Fourier analysis}.  We state here only a few examples. We refer the readers to \cite{QFA} and references therein for other examples and results.

\begin{theorem}[\bf Quantum Hausdorff-Young Inequality~\cite{JiLWu-1}]
Let $0\neq A$ be a two-box in an irreducible, subfactor planar algebra with quantum dimension  $\delta$. Then
\begin{equation}\label{eq:hysubfactor}
\|\FS(A)\|_q\leq\delta^{1-\frac{2}{p}}\|A\|_p\;, \quad 1\leq p\leq 2,\quad  p^{-1}+q^{-1}=1\;.
\end{equation}
The equality holds in  \eqref{eq:hysubfactor} for some $1<p<2$, iff the equality holds for all $1\leq p\leq 2$, iff $A$ is a  bi-shift of a biprojection.
\end{theorem}

Not only does this yield the optimal constant in the quantum inequality, but its extremizers are characterized by interesting algebraic properties.
This has been applied by ZL in the classification of subfactors in \cite{Liuex}. This gave rise to a quantum analogue of the fundamental theorem of finite abelian groups:  
An irreducible subfactor planar algebra generated by 2-boxes is a free product of finite abelian groups and Temperley-Lieb-Jones planar algebras, if and only if its 2-boxes are commutative, and its 3-boxes (modulo the basic construction ideal) are commutative.

\begin{theorem}[\bf Quantum Schur Product Theorem~\cite{Liuex}]
Let $A, B>0$ be two-box positive operators in a subfactor planar algebra. Then their convolution $A*B>0$. 
\end{theorem}
This theorem is a quantum generalization of the classical Schur product theorem that the Hadamard product of positive matrices is positive. 
Applying this result to the Drinfeld center of a unitary fusion category $\C$, we obtain the Schur product theorem on the Fourier dual of the Grothendieck ring $R$ of $\C$, namely irreducible representations of $R$. However, the Schur product theorem does not hold on a fusion ring in general. Therefore it provides an analytic obstruction to unitary categorification of a fusion ring; this obstruction turns out to be surprisingly efficient, as shown in \cite{LPW}. Other quantum inequalities have been investigated as analytic obstructions in \cite{LPW}, including Young's inequality and sum-set estimates.

There is a rich spectrum of analytic inequalities for quantum symmetries.  
There are also uncertainty principles that arise in case for certain $q,p$ one has equality, such as in the inequality \eqref{eq:hysubfactor} which holds as an identity for $p=q=2$.  A corresponding quantum Young's inequality also holds, with convolution defined by horizontal, rather than vertical, multiplication.  

Another interesting result that relates to quantum information is an entropic uncertainty principle for quantum entanglement.  There are several new inequalities involving Fourier duality, entropy, or  related uncertainty principles. We give an example here, that gained insights from picture language:
 The R\'enyi entropy of order $p$ for a two-box $A$ in a subfactor planar algebra is 
 	\be
	h_{p}(A)= \frac{p}{1-p} \,\log \norm{A}_{p}\;,
	\quad\text{for}\quad 0<p\neq1\;.
	\ee 
This tends to the von Neumann entropy as $p\searrow1+$:
\be
h(A)=Tr(-|A| \log |A|).
\ee

\begin{theorem}[\bf Quantum Entropic Uncertainty Principle~\cite{JiLWu-1}]
Let $A$ be a two-box in an irreducible, subfactor planar algebra with quantum dimension  $\delta$. Then
\be
h(|A|^{2}) + h(|\FS(A)|^{2}) \geq 2 \|A\|_2^2 \log \delta  -2 \|A\|_2^2 \log \|A\|_2^2\;.
\ee
\end{theorem}
\noindent
Quantum R\'enyi entropic uncertainty principles were proved in ~\cite{Entropy}, and the above result appeared as a special limit.

The von-Neumann entropy could be interpreted as an entanglement entropy between the domain and the range of an operator. The above result could be transformed into an uncertainty principle for entanglement entropy.

For example, the zero state and the Bell state are Fourier dual to each other. The above uncertainty principle yields a lower bound for the sum of their entanglement entropy.  As the zero state is a product state with no entanglement, the Bell state achieves  maximal entanglement.
Using the entropic uncertainty principle of the general SFT on n-boxes~\cite{JiLWu-1}, similar arguments apply to other states which are Fourier dual to the zero state, such as the GHZ or Max state mentioned above. From this point of view, we believe that the Fourier duality produces a {\it minimal-maximal entanglement pair}.

Following the pictorial intuition, ZL, Jiang and Wu introduced block maps for subfactors and proved that any 2-box converges to a multiple of a biprojection under the dynamic action of the block map \cite{JLWscm}. This can be considered as a quantum analogue of the central limit theorem, with biprojections playing the role of Gaussians. The block map could also be considered as a 2D renormalization flow from the physics point of view.

\section{Algebraic Identities}
Many other results have already come from the analysis of these pictures.  Some of the greatest insights seem to be associated with the three-dimensional quon language~\cite{JLW-Quon}.  This appears to be a new and very interesting field.  
ZL has used the quon language to establish a duality relation  leading to the proof of an interesting identity for $6j$-symbol self-duality.  This identity had been discovered for quantum $SU(2)$ in the context of the study of quantum gravity\cite{Barrett}.  But it was only conjectured in more general form~\cite{Shakirov}.  Here the $S$ matrix plays the role of $\FS$ \cite{LiuXu}.
\begin{theorem}[See \S6 of Liu~\cite{Quon-Surface}]
With $\overline{X}$ denoting the dual object to $X$ in a modular tensor category, 
\be\label{6jRelation}
\left|{{{X_{6}~X_{5}~X_{4}}\choose{\overline{X_{3}}~\overline{X_{2}}~\overline{X_{1}}}}}\right|^{2}
= \sum_{\vec {Y}} \left(\prod_{k=1}^{6}S_{X_{k}}^{Y_{k}} \right)
\left|{{{Y_{1}Y_{2}Y_{3}}\choose{Y_{4}Y_{5}Y_{6}}}}\right|^{2}\;.
\ee
 \end{theorem}

Usually the $S$ matrix is known, but the 6j-symbols are hard to solve. 
Even for quantum $SU(N)$, the closed form formula of $6j$ symbols remains unknown.
From the point of view of Fourier analysis, we suggest another scheme to compute the 6j symbols.
It is possible to compute parts of the 6j-symbols and then to solve the modulus of 6j-symbols from partially known data in Equation \eqref{6jRelation}, similar to  signal recovery \cite{CRT04,Donoho06}. One may further solve for the phases using the fact that the 6j-symbols form a biunitary.

\section{A de Finetti Theorem for Braids}
The famous de Finetti theorem in classical probability theory 
clarifies the relationship between permutation symmetry and the independence of a sequence of random variables~\cite{dF31}. Consequently an infinite sequence of symmetric random variables can be written as a convex combination (or integration) of independent identically
distributed (i.i.d.) random variables.

Motivated by pictures, we have established  a de Finetti theorem for states on  para{\-}fermion algebras of order $d$. In particular, we use the fact that a pair of parafermions of order $d$ generate the 
$d\times d$ matrix algebra $\mathbb{M}_d(\mathbb{C})$, that we denote by $PF_2$.  Thus it is natural to consider pairs of parafermions as a unit, and to study double braids that exchange these pairs.  The adjoint action of a double braid $b_{j}$ acting on strands $2j-1, \ldots, 2j+1$ implements  the second Reidermeister move in the form, 
\be
\raisebox{-1.5cm}{
\scalebox{.5}{\begin{tikzpicture}
\begin{scope}[shift={(4,-4)},yscale=1,xscale=1]
\node at (0, -.2) {$_{_{_{_{ 2j-1}}}}$};
\node at (1, -.2) {$_{_{_{_{ 2j}}}}$};
\node at (2, -.2) {$_{_{_{_{ 2j+1}}}}$};
\node at (3, -.2) {$_{_{_{_{ 2j+2}}}}$};
\draw (2,0)--(0,2);
\draw (3,0)--(1,2);
\drawWL {}{0,0}{2,2};
\drawWL {}{1,0}{3,2};
\end{scope}
\node at (3.8, -1.9){$_{m}$};
\node at (4.8, -1.1){$_{n}$};
\draw (4,-1)--(4, -2);
\draw (5,-1)--(5, -2);
\draw (6,-1)--(6, -2);
\draw (7,-1)--(7, -2);
\begin{scope}[shift={(4,1)},yscale=-1,xscale=1]
\node at (0, -.2) {$_{_{_{_{ 2j-1}}}}$};
\node at (1, -.2) {$_{_{_{_{ 2j}}}}$};
\node at (2, -.2) {$_{_{_{_{ 2j+1}}}}$};
\node at (3, -.2) {$_{_{_{_{ 2j+2}}}}$};
\draw (2,0)--(0,2);
\draw (3,0)--(1,2);
\drawWL {}{0,0}{2,2};
\drawWL {}{1,0}{3,2};
\end{scope}
\end{tikzpicture}}}
\ =\ 
\raisebox{-1.25cm}{\scalebox{.5}{\begin{tikzpicture}
\begin{scope}
\draw (4,3)--(4, -2);
\draw (5,3)--(5, -2);
\draw (6,3)--(6, -2);
\draw (7,3)--(7, -2);
\node at (5.8, 0.1){$_{m}$};
\node at (6.8, .9){$_{n}$};
\end{scope}
\end{tikzpicture}}}\  \  .
\ee
Similarly, the adjoint action of $b_{j}b_{j-1}$ moves the charges on four-strings:
\[
\scriptstyle{Ad(b_{j}b_{j-1})}\left(
\raisebox{-2.35cm}{\scalebox{.5}{\begin{tikzpicture}
\begin{scope}
\node at (2, -5.2) {$_{_{_{_{ 2j-2}}}}$};
\node at (3, -5.2) {$_{_{_{_{ 2j-3}}}}$};
\node at (4, -5.2) {$_{_{_{_{ 2j-1}}}}$};
\node at (5, -5.2) {$_{_{_{_{ 2j}}}}$};
\node at (6, -5.2) {$_{_{_{_{ 2j+1}}}}$};
\node at (7, -5.2) {$_{_{_{_{ 2j+2}}}}$};
\draw (2,4)--(2, -5);
\draw (3,4)--(3, -5);
\draw (4,4)--(4, -5);
\draw (5,4)--(5, -5);
\draw (6,4)--(6, -5);
\draw (7,4)--(7, -5);
\node at (1.8, -.6){$_{k}$};
\node at (2.8, -.6){$_{l}$};
\node at (3.8, -.6){$_{m}$};
\node at (4.8, -.6){$_{n}$};
\end{scope}
\end{tikzpicture}}}   \right)\ 
=\
\raisebox{-2.5cm}{
\scalebox{.5}{\begin{tikzpicture}
\begin{scope}[shift={(4,-4)},yscale=1,xscale=1]
\node at (0, -2.2) {$_{_{_{_{ 2j-1}}}}$};
\node at (1, -2.2) {$_{_{_{_{ 2j}}}}$};
\node at (2, -2.2) {$_{_{_{_{ 2j+1}}}}$};
\node at (3, -2.2) {$_{_{_{_{ 2j+2}}}}$};
\node at (-1, -2.2) {$_{_{_{_{ 2j-2}}}}$};
\node at (-2, -2.2) {$_{_{_{_{ 2j-3}}}}$};
\draw (-1, 0)--(1,-2);
\draw (-2, 0)--(0,-2);
\draw (2,0)--(0,2);
\draw (3,0)--(1,2);
\draw (2,0)--(2,-2);
\draw (3,0)--(3,-2);
\drawWL {}{0,0}{-2,-2};
\drawWL {}{1,0}{-1,-2};

\drawWL {}{0,0}{2,2};
\drawWL {}{1,0}{3,2};
\end{scope}
\node at (3.8, -1.4){$_{m}$};
\node at (4.8, -1.4){$_{n}$};
\node at (2.8, -1.4){$_{l}$};
\node at (1.8, -1.4){$_{k}$};

\draw (4,-1)--(4, -2);
\draw (5,-1)--(5, -2);
\draw (6,-1)--(6, -2);
\draw (7,-1)--(7, -2);
\draw (2,1)--(2, -4);
\draw (3,1)--(3, -4);

\begin{scope}[shift={(4,1)},yscale=-1,xscale=1]
\node at (0, -2.2) {$_{_{_{_{ 2j-1}}}}$};
\node at (1, -2.2) {$_{_{_{_{ 2j}}}}$};
\node at (2, -2.2) {$_{_{_{_{ 2j+1}}}}$};
\node at (3, -2.2) {$_{_{_{_{ 2j+2}}}}$};
\node at (-1, -2.2) {$_{_{_{_{ 2j-2}}}}$};
\node at (-2, -2.2) {$_{_{_{_{ 2j-3}}}}$};
\draw (-1, 0)--(1,-2);
\draw (-2, 0)--(0,-2);
\draw (2,0)--(0,2);
\draw (3,0)--(1,2);
\draw (2,0)--(2,-2);
\draw (3,0)--(3,-2);
\drawWL {}{0,0}{-2,-2};
\drawWL {}{1,0}{-1,-2};

\drawWL {}{0,0}{2,2};
\drawWL {}{1,0}{3,2};
\end{scope}
\end{tikzpicture}}}
\ =\ 
\raisebox{-2.45cm}{\scalebox{.5}{\begin{tikzpicture}
\begin{scope}
\node at (2, -5.2) {$_{_{_{_{ 2j-2}}}}$};
\node at (3, -5.2) {$_{_{_{_{ 2j-3}}}}$};
\node at (4, -5.2) {$_{_{_{_{ 2j-1}}}}$};
\node at (5, -5.2) {$_{_{_{_{ 2j}}}}$};
\node at (6, -5.2) {$_{_{_{_{ 2j+1}}}}$};
\node at (7, -5.2) {$_{_{_{_{ 2j+2}}}}$};
\draw (2,4)--(2, -5);
\draw (3,4)--(3, -5);
\draw (4,4)--(4, -5);
\draw (5,4)--(5, -5);
\draw (6,4)--(6, -5);
\draw (7,4)--(7, -5);
\node at (3.8, -.4){$_{k}$};
\node at (4.8, -.4){$_{l}$};
\node at (5.8, -.4){$_{m}$};
\node at (6.8, -.4){$_{n}$};
\end{scope}
\end{tikzpicture}}}\  \  .
\]

The infinite parafermion algebra $PF_{\infty}$ is a $\mathbb{Z}_d$-graded, tensor product of algebras $PF_{2}$ of parafermion pairs. 
Here we consider the braid group $\mathbb{B}_{\infty}$, acting on pairs of parafermions .  Let $S_{\mathbb{B}_{\infty}}$ denote the states on $PF_{\infty}$ that are invariant under the action of $\mathbb{B}_{\infty}$. Here are two results: 
\begin{theorem}[\textbf{First de Finitti  for braided parafermions}~\cite{BJLW}]\label{Thm:main1}
Let $\varphi \in S_{\mathbb{B}_{\infty}}$ be a braid-invariant state on  $PF_{\infty}$. 
Then the following are equivalent:
\begin{enumerate}
\item The state $\varphi$ is extremal in the set of states $S_{\mathbb{B}_{\infty}}$ on $PF_{\infty}$.
\item The state $\varphi=\rho^{\otimes\infty}$ is the infinite tensor product of a state $\rho$ on $PF_2$. 
\end{enumerate}
\end{theorem}
As a consequence, any  $\mathbb{B}_{\infty}$-invariant state on $PF_{\infty}$ is in the closure of the convex hull of the product states.
Let $\overline{PF}^{\varphi}_{\infty}$ denote the von Neumann algebra
generated by $PF_{\infty}$ in the Gel{\-}fand-Naimark-Segal (GNS) construction  with respect to the  state $\varphi\in S_{\mathbb{B}_{\infty}}$.  Also let   $(\overline{PF}^{\varphi}_{\infty})^{\mathbb{B}_{\infty}}$ be the fixed point algebra under the action of the braid group $\mathbb{B}_{\infty}$. 
The neutral subalgebra of  $(\overline{PF}^{\varphi}_{\infty})^{\mathbb{B}_{\infty}}$  is the subalgebra generated by monomials in parafermions of degree zero  mod $d$.  

It is interesting that a distinction arises in this characterization,  according to whether or not the order of the parafermion algebra is square free.    (This means that  $d=\prod_{i} p_i$\,, where the primes $p_i$ are distinct.)
Let us now suppose that the degree $d$ of the parafermion algebra is square free.  In this case one finds that extremal, braid-invariant states are neutral and that they  give rise to a factor.  One can refine Theorem \ref{Thm:main1}  as follows:

\begin{theorem}[\textbf{Second de Finitti  for square-free parafermion degree~\cite{BJLW}}]\label{Thm:main2}
Let $\varphi\in S_{\mathbb{B}_{\infty}}$ be a braid-invariant state on a parafermion algebra $PF_{\infty}$ of square-free degree $d$.  
The following are equivalent:
\begin{enumerate}
\item  The state $\varphi$ is extremal in $S_{\mathbb{B}_{\infty}}$.
\item  The state $\varphi= \rho^{\otimes\infty}$, where $\rho$ is  a neutral state on $PF_{2}$.
\item The neutral subalgebra of $(\overline{PF}^{\varphi}_{\infty})^{\mathbb{B}_{\infty}}=\mathbb{C}$.
\item The algebra $(\overline{PF}^{\varphi}_{\infty})^{\mathbb{B}_{\infty}}=\mathbb{C} $.
\item  The von Neumann algebra $\overline{PF}^{\varphi}_{\infty}$ is a factor.
\end{enumerate}
\end{theorem}
\noindent
We have also given a  characterization in case $d$ is not square free, see~\cite{BJLW}. 

\section{Some Outlooks}
We believe that we have only scratched the surface of an enormous new area of pictorial mathematics. We are excited by the fact that these advances have tentacles into many different subfields of mathematics.  They not only  have already made contact with algebra and topology.  But we also have found new and deep directions for  analysis appearing in our early results.  

In our paper~\cite{JL-Pic} we gave a long list of open questions.  Here are some complementary remarks on directions that we and our collaborators are investigating: 
\begin{enumerate}[I.]

\item {\bf Quantum Fourier Analysis.}  There are many new inequalities to be found, including topological Brascamp-Lieb inequalitites. There are new uncertainty principles still to be discovered.  The analysis of the period-$2n$ string Fourier transform on $n$-box transformations presumably has a rich and interesting  structure. 

\item{\bf Quantum Information.}  While we have already investigated the topological design of protocols, little work has been done on algorithms.  This will be a major focus in the future.  The quon language seems beautifully suited to describe stabilizer quantum error correction codes.  How to identify optimal codes is an open problem.  Also the use of quon language to understand questions of {\it topological complexity} appears very ripe.  
\end{enumerate}

\end{document}